\documentclass[a4paper,11pt,reqno]{amsart}
          \usepackage{amssymb}
	  \usepackage{amsmath}
	  \usepackage{amsthm}
          \usepackage{amsfonts}
          \usepackage[english]{babel}
          \usepackage[utf8]{inputenc}
          \usepackage{enumitem}

\usepackage[margin=2.5cm]{geometry}

\usepackage{tikz}
\usetikzlibrary{arrows,shapes}
\usepackage{caption}

 \usepackage[unicode,colorlinks,plainpages=false,hyperindex=true,bookmarksnumbered=true,bookmarksopen=false,pdfpagelabels]{hyperref}
 \hypersetup{urlcolor=cyan,linkcolor=blue,citecolor=red,colorlinks=true}
\newtheorem{thm}{Theorem}[section]
\newtheorem{cor}[thm]{Corollary}
\newtheorem{lem}[thm]{Lemma}
\newtheorem{prop}[thm]{Proposition}

\newtheorem{defn}[thm]{Definition}

\theoremstyle{definition}
\newtheorem{rem}[thm]{Remark}
\newtheorem*{conv*}{Conventions}
\newtheorem*{conte*}{Content}

\numberwithin{equation}{section}

\newcommand{\be}{\begin{equation}}
\newcommand{\ee}{\end{equation}}
\newcommand{\bea}{\begin{eqnarray}}
\newcommand{\eea}{\end{eqnarray}}
\newcommand{\ba}{\begin{aligned}}
\newcommand{\ea}{\end{aligned}}

 \usepackage{color}
\definecolor{lila}{rgb}{1,0.2,0.9}
\definecolor{purple}{rgb}{0.4,0,1}
\definecolor{darkgreen}{rgb}{0,0.5,0}

\def\span{\mathrm{span}}                    %
\def\cH{{\mathcal H}}                       %
\def\ri{{\rm i}}                            %
\def\bC{\mathbb{C}}                         %
\def\cF{{\mathcal F}}                       %
\def\dt {\left.\frac{d}{dt}\right|_{t=0}}   %
\def\rank{{\mathrm{rank}}}                  %
\def\Ad{\mathrm{Ad}}                        %
\def\cU{\mathcal{U}}                        %
\def\bR{\mathbb{R}}                         %
\def\fH{\mathfrak{H}}                       %
\def\fF{\mathfrak{F}}                       %
\def\ft{\mathfrak{t}}                       %
\def\fg{\mathfrak{g}}                       %
\def\ddim{{\mathrm{ddim}}}                  %
\def\cV{\mathcal{V}}                        %
\def\cJ{\mathcal{J}}                        %
\def\reg{{\mathrm{reg}}}                    %
\def\dind{{\mathrm{dind}}}                  %
\def\cA{\mathcal{A}}                        %
\def\reg{\mathrm{reg}}                      %
\def\ann{\mathrm{ann}}                      %
\def\Im{\mathrm{Im}}
\def\cO{\mathcal{O}}                        %
\def\fk{\mathfrak{k}}                       %
\def\fp{\mathfrak{p}}                       %
\def\diag{{\mathrm{diag}}}                  %
\def\SU{\mathrm{SU}}                        %
\def\mU{\mathrm{U}}                         %
\def\mS{\mathrm{S}}                         %
\def\fsu{\mathfrak{su}}                     %

\begin{document}

\title{Collective superintegrable systems
\\
 from the Guillemin--Sternberg torus action}

\maketitle

\begin{center}

L. Feh\'er${}^{a,b}$

\medskip
${}^a$Department of Theoretical Physics, University of Szeged\\
Tisza Lajos krt 84-86, H-6720 Szeged, Hungary\\
e-mail: lfeher@physx.u-szeged.hu

\medskip
${}^b$Institute for Particle and Nuclear Physics\\
HUN-REN Wigner Research Centre for Physics\\
 H-1525 Budapest, P.O.B.~49, Hungary

\end{center}

\begin{abstract}
We present a novel approach to the superintegrability of  collective Hamiltonians  invariant under
a Hamiltonian action of  a
compact connected
semisimple Lie group, $G$,
on a symplectic manifold, $M$.
 For this purpose, we consider $\fH := \cJ^*(C^\infty(\fg^*)^G)$
and the  centralizer $\fF$ of $\fH$ with respect to the Poisson bracket on $C^\infty(M)$,
where $\cJ: M \to \fg^*$ is the momentum map of the $G$-action.
By exploiting a Hamiltonian torus action that goes back to
Guillemin and Sternberg~[GS,1983], we demonstrate
that the functional dimensions of $\fH$ and $\fF$
satisfy the equality $\ddim(\fH) + \ddim(\fF) = \dim(M)$.
Together with a non-triviality condition, this ensures that the Abelian Poisson algebra
$\fH\subset C^\infty(M)$ represents a superintegrable system, and
 it also follows that the momentum map of the GS torus action
yields  action variables for the system.
Our work provides new insight into collective superintegrability
complementing earlier results of Bolsinov~and~Jovanovi\'{c}.
\end{abstract}

\setcounter{tocdepth}{2}

\bigskip

\tableofcontents

\newpage

 \section{Introduction}
 \label{sec1}

The study of classical and quantum superintegrability is motivated by the existence of many interesting examples and connections
to several areas of physics and mathematics. See, for example, the pioneering papers \cite{MF,Nek} and the reviews \cite{J,MPW,R}.
Here, we will be concerned with classical superintegrable systems on symplectic manifolds, adopting the following definition.

\begin{defn}\label{defnI}
A superintegrable system on a $C^\infty$ symplectic manifold $(M, \omega)$ is given by an Abelian Poisson subalgebra $\fH$ of $C^\infty(M)$ such that
all Hamiltonians $\cH\in \fH$ possess complete flows and the two relations
\be
1\leq \ddim(\fH) <  \dim(M)/2
\label{I1}\ee
and
\be
\ddim(\fH) + \ddim(\fF) = \dim(M)
\label{I2}\ee
hold for the functional dimensions of $\fH$ and its centralizer
 \be
\fF := \{ \cF\in C^\infty(M)\mid \{\cF,\cH\}=0,\,\, \forall \cH \in \fH\}
\label{I3}\ee
in the Poisson algebra $(C^\infty(M), \{-,-\})$ associated with the symplectic structure.
\end{defn}

It is implicit in the definition that $\fH$ and $\fF$ have well-defined functional
dimensions, that is,  the exterior derivatives of their elements
span subspaces of $T_p^*M$  of  dimension  $\ddim(\fH)$ and $\ddim(\fF)$, respectively,
for every point $p$ of a dense subset of $M$.
Locally around generic points, this means that
 $\fH$ is functionally generated by
$\ell := \ddim(\fH)$  independent Hamiltonians $\cH_1,\dots, \cH_\ell$ in involution
that admit $m:= \dim(M) - 2 \ell$  further joint constants of motion $\cF_1,\dots, \cF_m$ so that the
$(\ell + m)$ functions $\cH_1,\dots, \cH_\ell, \cF_1, \dots, \cF_m$ are independent and
functionally  generate $\fF$.
 The integer $\ell$ is  called the rank of the system.
 If the generic level surfaces of $\fF$ are compact, then their connected components are isotropic tori of dimension
  $\ell$, on which every $\cH\in \fH$ has quasi-periodic trajectories.
 Thus, the trajectories of the Hamiltonians $\cH\in \fH$ of superintegrable systems are more
 restricted than those of Liouville integrable systems, which can be defined similarly
 but with $\ddim(\fH) = \dim(M)/2$, in which case the equality \eqref{I2} is automatic.

 It is well known that
 every effective
 Hamiltonian action of the torus ${\mathrm{U}}(1)^\ell$ on a connected symplectic manifold $M$ with $\dim(M) > 2\ell$
  gives rise
 to a superintegrable  system \cite{Z}.  In fact, taking $\fH$ to be the Abelian Poisson algebra
 generated by the components of the momentum map  one can show, using some invariant theory,
 that $\ddim(\fH) + \ddim(\fF) = \dim(M)$ holds for the centralizer.

The first  goal of the current paper is to present a novel approach to  generalizing the above statement from torus actions
to actions of  non-Abelian compact Lie groups.
For this purpose, we shall consider  a connected symplectic manifold $(M,\omega)$
 equipped with a Hamiltonian action of a compact connected Lie group $G$ associated with a
 simple (or reductive)
 Lie algebra $\fg$.
By using the corresponding equivariant  momentum map \cite{OR}
\be
\cJ: M\to \fg^*,
\label{I4}\ee
we introduce
\be
\fH:= \cJ^*\left( C^\infty(\fg^*)^G\right).
\label{I5}\ee
Here,  $C^\infty(\fg^*)^G$ stands for the coadjoint invariants, i.e., the center of the Lie--Poisson algebra $C^\infty(\fg^*)$.
 The integral curves of the elements of $\fH$ are orbits of one-parameter subgroups of $G$, giving complete flows \cite{GS1}.
We shall demonstrate that \emph{the Abelian Poisson algebra $\fH$ \eqref{I5} and its centralizer $\fF$ \eqref{I3}  always satisfy
the crucial equality \eqref{I2}.}
Excluding the coadjoint orbits of $G$, for which $\ddim(\fH)=0$,
this general construction gives a superintegrable  system in the sense of Definition \ref{defnI}.
It can be applied, for instance, to the cotangent bundle $T^*Q$ of any $G$-manifold $Q$,  homogeneous spaces $Q=G/K$ being
important special cases.

We refer to the superintegrability  enjoyed by  $\fH$ \eqref{I5} as `collective superintegrability'.
This terminology is inspired by the seminal work of Guillemin and Sternberg \cite{GS1,GS2,GS3}
who called the functions of the form $h \circ \cJ$ with $h\in C^\infty(\fg^*)$ `collective Hamiltonians'.

Our proof of collective superintegrability will be based on the `Guillemin--Sternberg (GS) torus action',
which can be constructed
on a dense open $G$-invariant submanifold of every Hamiltonian $G$-manifold, where it commutes with
the original $G$-action.
Such a torus action  was first introduced in \cite{GS3} under the assumption that the generic coadjoint orbits in $\cJ(M)\subset \fg^*$
are regular orbits $G/T$, with $T$ being a maximal torus.
In the general case, which we learned about from  \cite{Lane,W}, the pertinent torus is given by the center of the
isotropy group of the intersection of a generic orbit of $G$ in $\cJ(M)$ with the
 fundamental Weyl chamber\footnote{This torus is denoted $T_\sigma$ below,
where $\sigma$ is the `principal open face' of the closed Weyl chamber $\ft_+$, for which
$\cJ^{-1}(\Ad_G(\sigma))$ is a dense open subset of $M$.
In  \cite{GS3} $\sigma$ was  the interior $\ft_+^\circ$, and then $T_\sigma = T$.}.
The GS torus action is Hamiltonian; its momentum map is obtained from $\cJ$ by
transforming it into the fundamental Weyl chamber via the coadjoint action of $G$.
 In combination  with Thimm's trick \cite{GS2,T},
it was used in \cite{GS3} to obtain  action variables for the classical Gelfand--Cetlin systems.
Here, we will show that \emph{it directly yields action variables for the collective superintegrable systems}
defined by $\fH$ \eqref{I5}.

It is a heuristic principle that torus actions lurk behind many features of integrable systems \cite{Z}.
Our second goal is to shed  light on the role of the GS torus action
in collective superintegrability, thereby illustrating the validity of the principle.

The superintegrability of collective Hamiltonians has been studied before by Bolsinov and Jovanovi\'{c} \cite{BJ},
who worked in a more general framework, allowing for proper Hamiltonian actions of non-compact Lie groups too.
In the paper \cite{BJ} and subsequent publications (see the review \cite{J}) they also gave several interesting applications.
Our proof of collective superintegrability amounts to a different approach, which applies to a subset
of the cases covered by \cite{BJ}.  However, our approach has the advantage of directly producing the action variables as well.

\medskip\noindent
{\bf Contents.}
In Section \ref{sec2} we collect the necessary  mathematical results
behind the GS torus action, which is described in Subsection \ref{subsec23}.
No originality is claimed here, although we also expound some technical details that we could not
find in the literature.
Subsection \ref{subsec24} is devoted to examples.
Then, we prove Theorem \ref{main} about collective superintegrability in Subsection \ref{subsec31},
which is  our main technical result, and
 present the related Theorem \ref{thmaction}  on generalized action-angle coordinates in Subsection  \ref{subsec32}.
A comparison with the results  of Bolsinov and Jovanovi\'{c} \cite{BJ} is given in Subsection \ref{subsec33}.
We conclude in Section \ref{sec4}, where we also discuss  potential future applications and a natural combination
of Theorem \ref{main} with Thimm's trick \cite{T}.

\section{Preparations}
\label{sec2}

In this section we collect useful mathematical results  from the literature \cite{DK,HNP,Lane,Mi,OR}.
Below, we  assume  that $G$ is a compact connected  Lie group with simple Lie algebra $\fg$,
but  the results hold, with straightforward  modifications, for semisimple and reductive compact Lie groups,  too.

\subsection{Properties of the fundamental Weyl chamber}
\label{subsec21}

It is convenient to realize $\fg$ as a real form of a complex simple Lie algebra $\fg_\bC$.
Then, a Cartan subalgebra $\ \ft_\bC$  of $\fg_\bC$ corresponds to the Lie algebra $\ft$ of a fixed maximal torus $T<G$.
The dual of the `real Cartan subalgebra' $ \ri \ft \subset \fg_\bC$ contains the set of roots $\Phi$ in which
we choose a base $\Delta$.  Using a Weyl--Chevalley basis \cite{DK,Sam}
 of $\fg_\bC$,  given by the co-roots $h_{\alpha_j}\in \ri \ft$ associated with the simple roots $\alpha_j \in \Delta$  and  `root vectors'
$e_{\pm \alpha}$ for $\alpha \in \Phi_+$,
 we can present $\fg$ as
\be
\fg = \span_\bR \{ \ri h_{\alpha_j}, \, e_\alpha^\pm \mid j=1,\dots, r,\,\, \alpha \in \Phi_+\}, \qquad r= \rank(\fg),
\label{T1}\ee
where
\be
e_\alpha^+ := \ri (e_\alpha + e_{-\alpha}),
\qquad
 e_\alpha^-:= (e_\alpha - e_{-\alpha}).
\label{T2} \ee

We identify $\fg^*$ and $\ft^*$ with $\fg$ and $\ft$, respectively, using the Killing form $\langle -, - \rangle$
normalized so that the long roots are of length $\sqrt{2}$.
As a result, the adjoint and coadjoint actions of $G$ are also identified, and will be denoted as $\Ad$.

We introduce the closed fundamental Weyl chamber $\ft_+$  by
\be
\ft_+ := \{ X \in \ft \mid \alpha_j(-\ri X) \geq 0,\, \forall j=1,\dots, r\}.
 \label{T3}\ee
Every subset $S$ of the simple roots determines a face $F_S$ of $\ft_+$ by setting
\be
F_S = \{ X \in \ft_+ \mid \alpha (-\ri X) = 0, \,\, \forall \alpha \in S \}.
\label{T4}\ee
The walls (facets) are associated with one-point subsets, i.e., with simple roots.
We let $F_S^\circ$ denote the relative interior of $F_S$, also called the algebraic interior \cite{HNP},
\be
F_S^\circ = \{ X\in F_S \mid \alpha(-\ri X) >0,\,\, \forall \alpha \in \Delta \setminus S\}.
\label{T5}\ee
From now on, we refer to $F_S^\circ$ as an `open face'.
In this notation, the open Weyl chamber $\ft_+^\circ$ can be written as $F^\circ_\emptyset$.
The boundary $\partial \ft_+$ decomposes as the disjoint union of the open faces for which $S\neq \emptyset$.
 For an illustration, see Figure \ref{Fig1}.

\begin{figure}
\includegraphics[width=5cm]{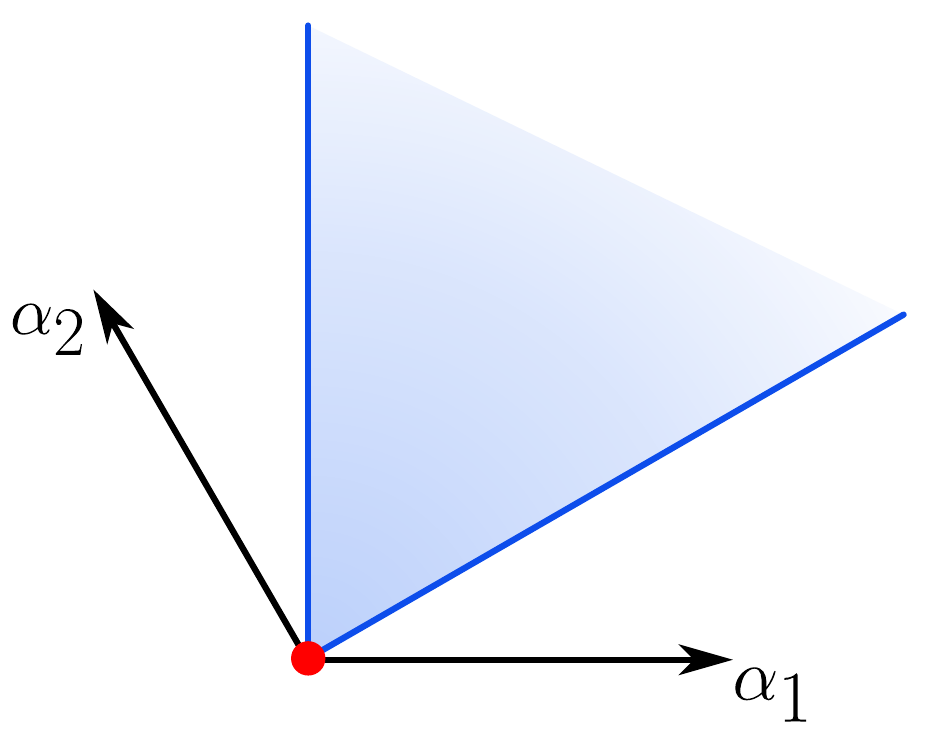}
\caption{The shaded blue region (extending to infinity) is the interior $\ft_+^\circ$ of the fundamental Weyl chamber $\ft_+$ of $\fsu(3)$.
The boundary $\partial \ft_+$ is the disjoint union of the origin and the blue open half-lines representing the relative interiors
of the one-dimensional faces orthogonal to the simple roots $\alpha_1$ and $\alpha_2$.}
\label{Fig1}
\end{figure}

The adjoint isotropy group of every $X \in \fg$ is connected.
Any two elements of an open face $\sigma=F_S^\circ$ have the same adjoint isotropy group,
which we denote by $G_\sigma$. Its Lie algebra $\fg_\sigma$ is generated, as a Lie algebra, by $\ft$
and the elements $e^\pm_\alpha$ for $\alpha \in S$. Equivalently, as a vector space:
\be
\fg_\sigma = \ft + \span_\bR \{ e_\alpha^\pm \mid \alpha \in \Phi_+,\, \alpha(-\ri X) = 0, \, \forall X \in \sigma\}.
\label{T6}\ee
For the center $\zeta(\fg_\sigma)$ we have the equality
\be
\zeta(\fg_\sigma) = \span_\bR (\sigma) =: \ft_\sigma,
\label{T7}\ee
and we let $T_\sigma$ be the  subgroup of $T$ with Lie algebra $\ft_\sigma$.
Note that
\be
\dim (T_\sigma) = \dim(T) - \vert S \vert ,
\label{T8}\ee
where $\vert S \vert$ is the number of elements of $S$. In particular, $G_\sigma = T$ for  $\sigma = \ft_+^\circ$.
For the properties summarized in this paragraph and in the next one, we refer to \cite{DK,HNP}.

The Weyl chamber $\ft_+$ is a fundamental domain of the Weyl group $W = N_G(T)/T$ acting on $\ft$.
The isotropy group $W_X$ of $X\in \sigma = F_S^\circ$ is constant along the open face, and thus it can be denoted by $W_\sigma$.
It is generated by the reflections in the simple roots $\alpha \in S$.
The elements of $\ft$ fixed by $W_\sigma$ are precisely the elements of $\ft_\sigma$.

\subsection{Invariant functions and the map to representatives}
\label{subsec22}

Consider a real $C^\infty$ function $F$ on an open subset $\cV \subset \fg$ and define its gradient $\nabla F\in C^\infty(\cV, \fg)$ by
\be
\langle \nabla F(Y_0), Y \rangle = \dt F(Y_0+ t Y), \qquad \forall Y_0\in \cV,\, \forall Y\in \fg.
\label{T9}\ee
Let $\zeta(\fg_{Y_0})$ denote the center of the centralizer  $\fg_{Y_0} < \fg$ of $Y_0$.

\begin{lem}\label{lemT1}
The function $F\in C^\infty(\cV)$ is invariant with respect to the infinitesimal action of $G$ on $\fg$,
i.e.~$F\in C^\infty(\cV)^\fg$,
 if and only if
$\nabla F(Y_0) \in \zeta(\fg_{Y_0})$ holds for every $Y_0 \in \cV$.
\end{lem}
\begin{proof}
The invariance condition on $F$ means that
\be
\langle \nabla F(Y_0), [Y, Y_0] \rangle = \langle [Y_0, \nabla F(Y_0)], Y \rangle = 0, \qquad \forall Y_0\in \cV,\, \forall Y\in \fg.
\label{T10}\ee
Thus, we obtain  $\nabla F(Y_0) \in \fg_{Y_0}$.
Taking any $Z,Y\in \fg$ and arbitrary but small enough real $t$ and $\tau$, it follows that
\be
F(\Ad_{e^{\tau Z}}(Y_0) + t Y ) = F (Y_0 +t \Ad_{e^{-\tau Z}}( Y) ),
\label{T11}\ee
which implies  the equality
\be
\nabla F(\Ad_{e^{\tau Z}} Y_0 ) =  \Ad_{e^{\tau Z}}(\nabla F(Y_0)).
\label{T12}\ee
Then, taking $Z \in \fg_{Y_0}$, we obtain that $\nabla F(Y_0) \in \zeta(\fg_{Y_0})$ for all $Y_0\in \cV$;
and if this holds, then obviously $F\in C^\infty(\cV)^\fg$.
\end{proof}

It is well known that the gradients of the $G$-invariant $C^\infty$ functions on $\fg$, and also the gradients
of the invariant polynomials,  span $\zeta(\fg_{Y_0})$ at every $Y_{0}\in \fg$.

Let us recall \cite{DK} that $\ft_+$ is not only a fundamental domain for the $W$-action on $\ft$, but also one for the adjoint action of $G$.
For any open face $\sigma=F_S^\circ$,  we define
\be
\Sigma_\sigma= G\cdot \sigma :=  \{ \Ad_g(X)  \mid X \in \sigma,\, g\in G\}.
\label{T13}\ee
Since the isotropy group is constant along $\sigma$, $\Sigma_\sigma$ is a connected component of an orbit type submanifold
for the adjoint action, and thus it is an embedded submanifold \cite{OR}.
Being a union of (co)adjoint orbits,
it is a Poisson submanifold as well with respect to the Lie--Poisson structure on $\fg\simeq \fg^*$.
The closure of this submanifold is $\bar \Sigma_\sigma = G \cdot \bar \sigma$.

Let us consider the map to representatives
\be
\Psi: \fg \to \ft_+, \qquad \Psi(Y) := \Ad_G(Y) \cap \ft_+,
\label{T14}\ee
which is sometimes called the `sweeping map' \cite{Lane}.  We shall also use the restricted maps
\be
\Psi_\sigma:= \Psi\vert_{\Sigma_\sigma}: \Sigma_\sigma \to \sigma,
\label{T15}\ee
and
\be
\psi: \ft \to \ft_+, \qquad \psi(X):= W\cdot X \cap \ft_+.
\label{T16}\ee

Now we recall an important Lie-theoretic result \cite{DK,HNP,Mi}.

\begin{thm}\label{thmT2}
The map $\Psi$ \eqref{T14} is continuous and descends to a homeomorphism from $\fg/G$ onto $\ft_+ = \ft/W$.
Its restriction $\Psi_\sigma: \Sigma_{\sigma} \to \sigma$ is $C^\infty$ for every open face.
\end{thm}

\begin{rem}\label{remT3}
The quotient topology of $\ft/W$ coincides with the topology of $\ft_+$ as a polyhedral cone.
The manifold structure of $\sigma$ is induced by its being
an open subset of a  linear subspace of $\ft$ of dimension $\dim(T_\sigma)$ \eqref{T8}.
For the second part, note that
the manifold $\Sigma_\sigma$ is equivariantly diffeomorphic to the trivial bundle $G/G_\sigma \times \sigma$
with $G$ acting only on the first factor.
The second part
is usually stated for the set of regular elements $\fg_\reg = \Sigma_{\ft_+^\circ}$ only.
It is also worth noting that the restricted map from $\bar \Sigma_\sigma$ to $\bar \sigma$
is continuous for every open face $\sigma$.
\end{rem}

Define the globally continuous real functions $\Psi_Z: \fg \to \bR$ and $\psi_Z: \ft \to \bR$ by
\be
\Psi_Z := \langle \Psi, Z \rangle
\quad \hbox{and}\quad  \psi_Z:= \langle \psi, Z \rangle, \qquad \forall Z\in \ft,
\label{T17}\ee
which are obviously $G$-invariant and $W$-invariant, respectively.
 In addition to the above-mentioned smoothness properties, we need some further details
 about these functions. To establish them, we start with two simple lemmas.

\begin{lem}\label{lemT4}
Pick $X_0\in   \partial \ft_+$.
Then, there exists a $W_{X_0}$-invariant open neighborhood $U_0 \subset \ft$ around  $X_0$ such that if
$w X\in \ft_+$ holds for some $X\in U_0$ and $w\in W$, then  $w \in W_{X_0}$.
\end{lem}
\begin{proof}
 Any $ w$ not fixing $X_0$ maps it
outside the closed Weyl chamber.
Consequently, if  $U_0$ is a small enough open ball centered on $X_0$,   then $w U_0 \cap \ft_+ =\emptyset$ for each
$w\notin W_{X_0}$.  It follows that $U_0$ has the
claimed property: If  $X \in U_0$ and $w X \in \ft_+$, then $w \in W_{X_0}$.
The $W_{X_0}$-invariance  of $U_0$ holds since the action of the Weyl group is orthogonal on $\ft$ with respect
to the Euclidean structure
defined by the negative of the Killing form.
\end{proof}

\begin{rem}\label{remT5}
It is worth noting that the open set $U_0$ of Lemma \ref{lemT4} cannot intersect the boundary $\partial\sigma$, since
for $X_0' \in \bar \sigma \setminus \sigma$ the group $W_{X_0}$ is a proper subgroup of $W_{X_0'}$.
\end{rem}

 \begin{lem}\label{lemT6}
 Pick $X_0 \in  \sigma $ from an open face $\sigma \subset \partial \ft_+$. Then,
 there exists a $W$-invariant open neighborhood $U\subset \ft$ of $X_0$ on which $\psi_Z$
 is $C^\infty$ for every $Z\in\ft_\sigma$.
 \end{lem}
 \begin{proof}
  Remember that $\ft_\sigma$   is the fixed point set of $ W_{X_0} =W_\sigma$.
 We know that $\psi(X)$ is given by $\psi(X) = w_X X$, where $w_X \in W$ is any element that transforms $X\in \ft$ into $\ft_+$.
 Taking $U_0$ to be the
 open neighborhood of $X_0$  exhibited in Lemma \ref{lemT4}, we have
   \be
   \psi_Z(X) =\langle Z, w_X X\rangle =  \langle w_X^{-1} Z, X\rangle = \langle Z, X\rangle, \quad \forall X\in U_0.
   \label{psiXZ}\ee
   Thus, the restriction of $\psi_Z$ to $U_0$ is $C^\infty$.
 Since $\psi_Z$ is a $W$-invariant function on $\ft$, it is also $C^\infty$  on $U = W\cdot U_0$.
 \end{proof}

It is plain that every $W$-invariant function on $\ft$ extends to a unique $G$-invariant function on $\fg$.
According to the classical Chevalley extension theorem,  every Weyl invariant polynomial on $\ft$ extends
to a unique $G$-invariant polynomial
on $\fg$,  and thus the polynomial rings
$\bR[\ft]^W$ and $\bR[\fg]^G$ are isomorphic.
It is known that the analogous  natural isomorphism holds for continuous as well as for $C^\infty$
functions (see e.g.~\cite[Cor.~3.29 and Thm. 30.31]{Mi}).
This readily implies the following `localized' extension property.

\begin{lem}\label{lemT7}
Take a continuous $W$-invariant function $f\in C(\ft)^W$ that restricts  to an element of $C^\infty(\cU)^W$ for
a $W$-invariant open subset $\cU \subset \ft$.
Then, the unique extension  $F\in C(\fg)^G$ is $C^\infty$ on the $G$-invariant open subset $\cV := \Ad_G(\cU)$ of $\fg$.
\end{lem}

The next result will play an important role in what follows.

 \begin{prop}\label{propT8}
 Choose $X_0 \in \sigma$, where $\sigma \subset \ft_+$ is an open face.
 Then, there exists a $G$-invariant
 open neighborhood $V$ of $X_0$ on which the function $\Psi_Z$ is $C^\infty$ for every $Z\in \ft_\sigma$.
 Moreover, if $Y_0 \in V$ is such that $X_0 = \Psi(Y_0) = \Ad_{g(Y_0)}(Y_0)$ for some $g(Y_0)\in G$, then
 the derivative of $\Psi_Z$ at $Y_0$ is given by
 \be
 \nabla \Psi_Z (Y_0) = \Ad_{g(Y_0)^{-1}}(Z).
 \label{T19}\ee
 \end{prop}
 \begin{proof}
 Consider the $W$-invariant neighborhood $U$ of $X_0$ exhibited in Lemma \ref{lemT6}.
 The unique $G$-invariant extension of $\psi_Z$ is obviously $\Psi_Z$, and Lemma \ref{lemT7} ensures
 that $\Psi_Z$ is $C^\infty$ on $V:= \Ad_G(U)$.
 Since $\Psi_Z$ is $G$-invariant on $V$, its derivative enjoys the relation
 \be
 \nabla \Psi_Z (Y_0) = \Ad_{g(Y_0)^{-1}}(\nabla \Psi_Z(X_0)).
 \label{T20}\ee
 The crux is that the derivative $\nabla \Psi_Z(X_0)$ is $\ft_\sigma$-valued, since $\ft_\sigma$ is the center of
 $\fg_{X_0} = \fg_\sigma$.
 Therefore $\nabla \Psi_Z(X_0)$ can be determined by  calculating
 the derivative at $X_0$ of the restriction of $\Psi_Z$ to a neighborhood $U_0$ of $X_0$ in $\ft$, i.e.,
 by calculating the derivative of $\psi_Z$ at $X_0$.
 Then, by using $U_0$ from Lemma \ref{T4},
 the claim  follows from equation \eqref{psiXZ}.
  \end{proof}

It is worth noting that $g(Y_0)$ can be replaced by $\eta g(Y_0)$ with any $\eta \in G_\sigma$,
 and this ambiguity drops out of the formula \eqref{T19} since $Z\in \ft_\sigma =\zeta(\fg_\sigma)$.

\subsection{The Guillemin--Sternberg torus action}
\label{subsec23}

The following result is stated in this form in \cite[Thm.~2 and Prop.~1]{Lane}.
It is a consequence of Lemmas 6.7--6.9 in \cite{HNP}. See also \cite[Thm.~3.1]{LMTW}.

\begin{thm}\label{thmT9}
Suppose that $\cJ: M \to \fg \simeq \fg^*$ is the equivariant momentum map for a smooth Hamiltonian action
of a compact connected Lie group, with simple Lie algebra $\fg$,   on a connected symplectic manifold $(M,\omega)$.
Then, there is a unique face $F_S$ \eqref{T4} of $\ft_+$ for which $\cJ(M) \cap \ft_+ \subseteq F_S$
and $\cJ(M) \cap F_S^\circ \neq \emptyset$.
For this open face $F_S^\circ=:\sigma$,
\be
M_\sigma:= \cJ^{-1}(\Sigma_\sigma) = G\cdot \cJ^{-1}(\sigma)
\label{T21}\ee
is a  $G$-invariant, dense open, connected submanifold of $M$.
\end{thm}

\begin{defn}\label{defnT10}
Following \cite{LMTW},
we call $\sigma$ appearing in Theorem \ref{thmT9} the `principal open face' of the $G$-action.
\end{defn}

Using the principal open face  $\sigma$,  we
introduce the restricted momentum map
\be
\cJ_\sigma: M_\sigma \to \fg.
\label{T22}\ee
It takes its values in the Poisson submanifold $\Sigma_\sigma \subset  \fg$ \eqref{T13},
but  we can also regard it as a map into $\fg$.

The next lemma is a direct consequence of standard properties of  equivariant momentum maps \cite{GS1}.

\begin{lem}\label{lemT11}
Let $F$ be a $G$-invariant real function on $\fg\simeq \fg^*$ that is $C^\infty$ on a $G$-invariant open neighborhood of $\Sigma_\sigma$
and consider $\cF := F \circ \cJ_\sigma \in C^\infty(M_\sigma)^G$.
Then, the integral curve $p(t)$ of the Hamiltonian vector field of $\cF$ through an arbitrary initial value $p_0 \in M_\sigma$
is given by the formula
\be
p(t) = \exp\left( t \nabla F(\cJ(p_0))\right) \cdot p_0,
\label{ptgen}\ee
where  $g\cdot p_0$ is used to
denote the action of any $g\in G$ on $p_0\in M$.
\end{lem}
\begin{proof}
The statement is a special case of the characterization of the integral curves of collective Hamiltonians (elements of
$\cJ^*(C^\infty(\fg^*)$)
developed in  \cite[pp.~241-242]{GS1}. It causes no trouble that the function $F$ is not necessarily globally smooth,
since the images of the integral curves of $\cF$  under the momentum map
belong to the open subset of $\fg^*$ where $F$ is smooth.
\end{proof}

\begin{cor}\label{corT12}
For any $Z\in \ft_\sigma$, the integral curve of the function $\Psi_Z \circ \cJ_\sigma \in C^\infty(M_\sigma)^G$ through
an arbitrary
initial value $p_0\in M_\sigma$ is given by the formula
\be
p(t) = \left(g(\cJ(p_0))^{-1} e^{t Z } g(\cJ(p_0)) \right)\cdot p_0,
\label{ptZ}\ee
where $g(\cJ(p_0)) \in G$ is any group element that satisfies $\Psi(\cJ(p_0))= \Ad_{g(\cJ(p_0))}(\cJ(p_0))$.
\end{cor}
\begin{proof}
We see from Proposition \ref{propT8}  that $\Psi_Z$ is $C^\infty$ on an open $G$-invariant neighborhood of $\Sigma_\sigma$,
since such a neighborhood is provided by the union of the $G$-invariant open sets  containing $X_0\in \sigma$ that we exhibited
in Proposition \ref{propT8} for every $X_0\in \sigma$. The statement is then verified by combining \eqref{T19} with \eqref{ptgen}.

Let us note in passing that if $\sigma= \ft_+^\circ$ then $\Sigma_\sigma$ is the open set of regular elements of $\fg$, which
in this case coincides with
 the open neighborhood mentioned above.
\end{proof}

The torus action  described below goes back to the paper \cite{GS3}, where its combination with Thimm's trick \cite{GS2,T} was
used to construct action variables for the Gelfand--Cetlin systems.
We note that the relations \eqref{T7} and \eqref{T15} are used in \eqref{musi}.

\begin{thm}\label{thmT13}
The smooth map
\be
\mu_\sigma:= \Psi_\sigma \circ \cJ_\sigma:
M_\sigma \to \ft_\sigma
\label{musi}\ee
is the momentum map for the Hamiltonian action of the torus $T_\sigma$ on $M_\sigma$ given by
\be
\tau_\sigma \star p_0 = (g(\cJ(p_0))^{-1} \tau_\sigma g(\cJ(p_0)))\cdot p_0,
\qquad \forall \tau_\sigma \in T_\sigma,\, \forall p_0 \in M_\sigma,
\label{GSact}\ee
where the dot refers to the original $G$-action and $g(\cJ(p_0))$ was defined in Corollary \ref{corT12}.
\end{thm}
\begin{proof}
The map $\mu_\sigma$ is a Poisson map into $\ft_\sigma \simeq \ft_\sigma^*$ since
it is a composition
of two Poisson maps. Indeed, $\cJ_\sigma$ is a Poisson map from $M_\sigma$ to the Poisson submanifold
$\Sigma_\sigma$ of $\fg \simeq \fg^*$.
The component functions $\langle  Z, \Psi_\sigma \rangle$ of the map $\Psi_\sigma$ Poisson commute
since they are $G$-invariant functions on the Poisson submanifold $\Sigma_\sigma$ of $\fg \simeq \fg^*$,
and thus are in the center of the Poisson bracket on $C^\infty(\Sigma_\sigma)$.
Observe from \eqref{ptZ} that the Hamiltonian vector fields of the functions  $ \langle Z, \mu_\sigma \rangle = \Psi_Z \circ \cJ_\sigma $ are complete.
We also see from the formula \eqref{ptZ} that these integral curves are $2\pi$-periodic for any $Z$ from the
integral lattice, i.e., from the kernel of the homomorphism from $\ft_\sigma$ to $T_\sigma$  given by
$Z \mapsto \exp(2\pi Z)$.  Therefore $\mu_\sigma$ is the momentum map for a Hamiltonian $T_\sigma$-action,
and the formula \eqref{ptZ} implies that this action operates according to \eqref{GSact}.
\end{proof}

\begin{rem}\label{remT14}
We refer to the action  \eqref{GSact} as the \emph{GS action} of $T_\sigma$.
As was already noted in Section \ref{sec1},
Guillemin and Sternberg assumed in \cite{GS3} that $\sigma = \ft_+^\circ$.
The generalization to an arbitrary
principal open face
$\sigma$ appears, for example, in the papers \cite{Lane,W}.
 In this respect,  \cite[Prop.~2]{Lane} refers to \cite[Prop.~3.4]{W} for a proof, while \cite{W} refers to \cite{W2}.
In fact, \cite[Prop.~8.1]{W2} provides a proof of the key formula \eqref{ptZ}, which is different from our proof.
(Their proof
is based on the fact that $\cJ^{-1}(\sigma)$ is a symplectic submanifold of $M$  on which $\langle Z, \mu_\sigma \rangle = \langle Z, \cJ\rangle$
holds, and the Hamiltonian vector fields of these functions are tangent to $\cJ^{-1}(\sigma)$.)
It may happen in examples (see below) that for some $Z\in \ft_\sigma$ the function $\Psi_Z \circ \cJ_\sigma$ is identically zero or constant on $M_\sigma$.
Of course, in such cases one has $p(t) = p_0$ for the integral curves of $\Psi_Z \circ \cJ_\sigma$.
The formula \eqref{ptZ}  is still valid, as one can also verify by the proof
presented in \cite{W2}.

It is  possible to obtain similar torus actions from $G$-actions generated by
group-valued momentum maps in the frameworks of Poisson--Lie groups and
quasi-Hamiltonian/quasi-Poisson geometry, too.
In fact, examples of such torus actions were used in \cite{FF} for studying some superintegrable systems.
\end{rem}

\begin{rem}\label{remT15}
The  $T_\sigma$-action \eqref{GSact}  commutes with the $G$-action on $M_\sigma$,
since its momentum map $\mu_\sigma$ is invariant under the $G$-action.
Of course, one can also check directly that the formula \eqref{GSact}  defines a smooth $T_\sigma$-action and
that it commutes with the original $G$-action on $M_\sigma$.
 The examples exhibited below show the $T_\sigma$-action is \emph{not} effective in general.
\end{rem}

\subsection{Illustration by examples}
\label{subsec24}

We wish to illustrate some features of the GS torus action \eqref{GSact} by examples.
Perhaps the simplest example is the cotangent bundle $T^*G$ with the $G$-action coming from left-multiplications.
Then, $\sigma= \ft_+^\circ$  and the GS action of $T_\sigma = T$ is free.
One may also equip  $T^*G$ with the $G$-action given by the cotangent lift of the
conjugation action of $G$ on itself.   One can show that the corresponding momentum map
is  surjective, and thus $\sigma = \ft_+^\circ$. In this case,  the GS action is trivial
precisely for the elements of the center $\zeta(G) < T$. These statements are not completely obvious, but we here omit the proof.

Note that the GS action of $T_\sigma$ is trivial if $M$ is a coadjoint orbit of $G$ with
its natural symplectic form. Every coadjoint orbit contains a unique representative $X_0 \in \ft_+$, in some
open face $\sigma$. Then, the momentum map $\mu_\sigma$ is just the constant map, $\mu_\sigma(X) = X_0$ for all $X\in \cO$.
A constant momentum map cannot generate non-trivial action.
It is easy to check that the formula \eqref{GSact} is still valid as it gives $\tau \star X = X$ for all $\tau \in T_\sigma$ and $X\in \cO$.

More interesting properties of the GS action can be seen by considering cotangent bundles of homogeneous spaces.
Taking  a connected subgroup $K < G$  with Lie algebra $\fk < \fg$, one can realize $T^*(G/K)$
as a symplectic reduction of $T^*G$ by  the right-multiplication action of $K$.
This gives the model
\be
T^*(G/K) = (G \times \fp)/K  = \{ [(\eta, L)] \mid (\eta, L) \in G \times \fp \},
\label{B1}\ee
where $\fp$ is the orthogonal complement of $\fk$ in $\fg$
and
\be
[(\eta, L)]  = \{ (\eta y^{-1}, \Ad_y(L)) \mid y\in K \}
\label{B2}\ee
with $\Ad$ denoting the adjoint action of $G$ on $\fg$.
Indeed, by using left-trivialization and the identification $\fg^* \simeq \fg$, $G\times \fp$ represents
the inverse image of the zero momentum map value for the $K$-action on $T^*G = G\times \fg$.
The action of $G$ on $T^*(G/K)$ and the corresponding momentum map read
\be
g \cdot [(\eta, L)] = [(g\eta, L)]
\quad \hbox{and} \quad \cJ([(\eta, L)]) = \Ad_\eta(L).
\label{B3}\ee

Turning to a specific example,  consider $G=\SU(3)$ and its block-diagonal subgroup $K = \mS(\mU(2) \times \mU(1))$.
Now, the elements of $\fp$ are the matrices of the form
\be
L := L(v):= \ri \begin{bmatrix} 0_2 & v \\
                            v^\dagger & 0
\end{bmatrix},\quad \forall v \in \bC^2,
\label{B4}\ee
where $0_2$ denotes the 2 by 2 zero matrix.
The fundamental Weyl chamber consists of the purely imaginary, traceless  diagonal matrices
\be
D = \ri \operatorname{diag}(d_1,d_2, d_3) \quad \hbox{with}\quad  d_1 \geq d_2 \geq d_3.
\label{B5}\ee
We have $\Psi(\Ad_\eta(L)) = \Psi(L)$ for every $\eta \in G$ and $L\in \fp$.
For $v\neq 0$, straightforward  diagonalization of $L(v)$ yields
\be
\Psi( L(v)) = g(v) L(v) g(v)^{-1} = \ri \sqrt{v^\dagger v} \operatorname{diag} (1, 0, -1),
\label{B6}\ee
where one may take
\be
g(v)^{-1} =  \begin{bmatrix} {v}/{\sqrt{2 v^\dagger v}} & u & {v}/{\sqrt{2 v^\dagger v}}\\
                  {1}/{\sqrt{2}} & 0 & -{1}/{\sqrt{2}} \end{bmatrix}
\label{B8}\ee
using $u\in \bC^2$ for which $u^\dagger u = 1$ and $u^\dagger v =0$. By choosing $u$, $\det(g(v)) =1$ can also be ensured.

We see from the above that the principal open face is $\sigma = \ft_+^\circ$,
but the image of the GS momentum map $\mu_\sigma$ is only the open half-line
\be
\{ \ri \, \diag(r,0,-r) \mid r>0 \} \subset \ft_+^\circ.
\label{B9}\ee
As a result,
 the momentum map component
$\Psi_Z = \langle Z, \mu_\sigma \rangle $ is \emph{identically zero}
for
\be
Z:= \ri \operatorname{diag} (1, -2, 1) \in \ft = \ft_\sigma.
\label{B10}\ee
This implies that the restriction of the GS action of $T= T_\sigma$
becomes trivial upon restriction to the
one-parameter subgroup $\{ e^{t Z} \mid t \in \bR\}$.
The formula \eqref{GSact} confirms this by virtue of the property
\be
g(v)^{-1} e^{tZ} g(v) \in K,
\label{B11}\ee
which is readily verified. In fact, if $g(v)^{-1} \tau g(v) \in K$, then one obtains
\be
\tau \star [(\eta, L(v))] = \eta \cdot (\tau\star [(e, L(v))]) = \eta \cdot [ (g(v)^{-1} \tau g(v), L(v)) ] = [(\eta, L(v))].
\label{B12}\ee
Here, we used that the $G$-action and the $T_\sigma$-action commute and in the last step also used that
\be
[ (g(v)^{-1} \tau g(v), L(v)) ]  = [(e, \Ad_{g(v)^{-1} \tau g(v)} L(v))] = [(e, L(v))],
\label{B13}\ee
which holds since $\Ad_{g(v)} L(v) \in \ft_+^\circ$ by the definition of $g(v)$.
A further inspection along these lines proves the following statement.

\begin{lem}\label{lemT16}
In the above $G=\SU(3)$ example the kernel of the GS action \eqref{GSact} is provided by the elements $\tau \in T$ for which
$g(v)^{-1} \tau g(v) \in K$ holds. This implies
 that \eqref{GSact} gives the trivial action  if and only if $\tau = \exp(tZ)$ with $Z$ \eqref{B10} and some $t\in \bR$.
\end{lem}

A generalization of the example is furnished by $G = \SU(n)$ and $K = \mS(\mU(n-1) \times \mU(1))$ for any $n > 3$.
One finds that the principal open face is now $\sigma=  F_S^\circ$ with $F_S$ \eqref{T4} defined by the subset
$S = \Delta \setminus \{ \alpha_1, \alpha_{n-1}\}$
of the usual simple roots.
The image of $\mu_\sigma$ is the open half-line
\be
\{ \ri \operatorname{diag}(r,0,\dots, 0, -r)\mid r>0\}
\ee
inside the  2-dimensional open face $\sigma$.

One may  use the arguments leading to Lemma \ref{lemT16} to characterize the kernel of the GS torus action
on cotangent bundles of homogeneous spaces,  $M= T^*(G/K)$, in general.  In fact, the action of $\tau \in T_\sigma$ is trivial
if and only if $g(L)^{-1} \tau g(L) \in K$ holds for every $[(e,L)]\in M_\sigma$, where $g(L)\in G$ is subject to $\Ad_{g(L)} (L)\in \sigma$.
Since $\mu_\sigma([(e,L)]) = \Ad_{g(L)} (L)$,
it follows that the one-parameter subgroup generated by $Z\in \ft_\sigma$ is in the kernel if and only if
$\langle Z, \mu_\sigma \rangle$ vanishes identically.
Equivalently, this means that $Z\in \ft_\sigma$ is perpendicular to every point in the image
of the map $\Psi \circ \cJ: M \to \ft_+$.  The interested reader can easily verify these claims.

\section{Collective superintegrability and action variables}
\label{sec3}

Now we come to the application of the results collected in Section \ref{sec2}.
We first show that the Abelian Poisson algebra $\fH$ \eqref{I5} yields a superintegrable system
and then explain how its action variables arise directly from the momentum map $\mu_\sigma$ \eqref{musi}.

\subsection{The proof of superintegrability}
\label{subsec31}

\begin{defn}\label{defnH1}
With $\sigma\subset \ft_+$ denoting the principal open face of the $G$-action as in Definition \ref{defnT10},
let $M_\sigma^* \subset M_\sigma$ be the principal orbit type submanifold \cite{DK, Mi} of $M_\sigma$ with
respect to the  $T_\sigma$-action \eqref{GSact}.
\end{defn}

Note that $M_\sigma^* \subset M_\sigma \subset M$ is a chain of dense open connected
submanifolds, and thus $M_\sigma^* \subset M$ is also a dense open connected submanifold.
We can form the smooth, connected quotient manifold
$M_\sigma^*/T_\sigma$, and, to avoid a trivial situation, we  assume that
\be
 \dim(M_\sigma^*/ T_\sigma) < \dim(M_\sigma^*).
 \label{dimcond1}\ee
This means that the  generic orbits of the $T_\sigma$-action \eqref{GSact}
 have positive dimension.
This condition excludes the case in which $M$ is a coadjoint orbit of $G$.
This also ensures that
\be
\ddim(\fH) < \frac{1}{2} \dim(M),
\label{dimcond2}\ee
 since  $\ddim(\fH) \leq \ddim( C^\infty(\fg^*)^G)=\rank(G)$
and
$2\,\rank(G)$ is smaller than the minimal dimension of any symplectic manifold on which
$G$ can act  non-trivially, except for the minimal coadjoint orbits of $\mathrm{SU}(n)$.
To see this, it is enough to notice that $\dim(G\cdot x) \geq \dim(G\cdot \cJ(x))$, and recall
the well-known fact
that $\dim(G \cdot \xi) \geq 2\rank(G)$ for any nonzero $\xi \in \fg\simeq \fg^*$, except for the minimal
orbits $\mathbb{C P}^{n-1} \subset \mathfrak{su}(n)$.

\begin{thm}\label{main}
Consider a Hamiltonian action of a compact connected  Lie group $G$ with simple Lie algebra $\fg$
on a connected symplectic manifold $(M,\omega)$, and let  $\cJ: M\to \fg\simeq \fg^*$ be the momentum map.
Then, the Abelian Poisson algebra $\fH$ \eqref{I5} of $G$-invariant collective Hamiltonians
and its centralizer $\fF$ \eqref{I3} in $C^\infty(M)$  satisfy the relation
$\ddim(\fH) + \ddim(\fF) = \dim(M)$.
Consequently, assuming the non-triviality condition \eqref{dimcond1},
 one obtains a superintegrable system in the sense of Definition \ref{defnI}.
\end{thm}

We shall obtain Theorem \ref{main} by calculating the functional dimensions of $\fH$ and $\fF$ separately.

\begin{lem}\label{lemH3}
The functional dimension of $\fH$ is given by
\be
\ddim(\fH) = \dim(M) - D
\quad \hbox{with}\quad
D:=  \dim(M_\sigma^*/ T_\sigma).
\label{ddimH}\ee
\end{lem}
\begin{proof}
For any $H\in C^\infty(\fg)^G$, we have
\be
H \circ \cJ = H \circ \Psi \circ \cJ,
\label{H4}\ee
and after restriction to the dense open submanifold $M_\sigma^* \subset M$ the map $\Psi \circ \cJ$
becomes the momentum map of the action \eqref{GSact} of $T_\sigma$.
 This momentum map, $\mu_\sigma = \Psi_\sigma \circ \cJ_\sigma$, has $\dim(M) - D$ functionally independent components
on $M_\sigma^*$, since this is the dimension of the generic orbits of the $T_\sigma$-action.
We obtain immediately that
\be
\ddim(\fH) \leq  \dim(M) -D.
\label{H5}\ee
To prove the equality, consider the restriction of the functions $\Psi_Z$ to $\Sigma_\sigma$ \eqref{T13}
for every $Z\in \ft_\sigma$.
These restricted functions belong to $C^\infty(\Sigma_\sigma)^G$.
It follows from standard extension properties of invariant functions under proper actions \cite[Props.~2.5.6 and 2.5.7]{OR}
that for any $Y_0 \in \Sigma_\sigma$ there exists a $G$-invariant open subset $U$ of $\Sigma_\sigma$ containing $Y_0$
and a function $H\in  C^\infty(\fg)^G$
that coincides with $\Psi_Z$ on $U$.
This implies that the restriction of $\fH$ to $M_\sigma$ has at least the same functional
dimension as the algebra of functions generated by $\mu_\sigma$ \eqref{musi}.
Taking into account that $M_\sigma^* \subset M_\sigma \subset M$ are dense open submanifolds,
the  proof is complete.
\end{proof}

The proof of Lemma \ref{lemH3} also shows that the following statement is valid.

\begin{cor}\label{corH4}
At every point of $M_\sigma^*$,
the span of the differentials of the elements of $\fH$ is the same as the span of
the differentials of the components of the $T_\sigma$ momentum map $\mu_\sigma$ \eqref{musi}.
By continuity, the same holds on $M_\sigma$.
\end{cor}

For any natural number $j$ and constant $r>0$,  let $B^j_r \subset \bR^j$ denote the
open ball of radius $r$ around the origin.

\begin{lem}\label{lemH5}
Maintaining the previous notations and assumptions,
the centralizer $\fF$ \eqref{I3} of $\fH$ \eqref{I5} satisfies
$\ddim(\fF) = D$ with $D$ in \eqref{ddimH}.
\end{lem}
\begin{proof}
Let $\pi: M_\sigma^* \to M_\sigma^*/T_\sigma$ be the canonical projection
and choose an arbitrary  point $x_0\in M_\sigma^*$.  With $y_0:=\pi(x_0)$,
introduce a coordinate system $(U, \chi_1,\dots, \chi_D)$ on $M_\sigma^*/T_\sigma$ that
is centered on $y_0$ and the image of the coordinate map $\chi:= (\chi_1,\dots, \chi_D): U \to \bR^D$
is an open ball $B_1^D$ of radius $1$.  Then, one can find functions $f_i \in C^\infty(M_\sigma^*/T_\sigma)$ that
coincide with the $\chi_i$ on $\chi^{-1}(B_{1/3}^D)$ and are identically zero outside $\chi^{-1}(B_{2/3}^D)$.
By using such functions $f_i$,
let us define  the functions $\cF_i$ on $M$ by setting them equal to
$\pi^*(f_i)$ on $M_\sigma^*$ and defining them to be identically zero outside $M_\sigma^*$.
It is clear that this gives $\cF_i\in C^\infty(M)$ and the $D$ functions are independent
at $x_0$ (actually also on $\pi^{-1}(\chi^{-1}(B_{1/3}^D))$).
The functions $\cF_i$  Poisson commute with the component functions of the momentum map $\mu_\sigma$
on $M_\sigma^*$, since their restrictions to $M_\sigma^*$ are $T_\sigma$-invariant.
Thus, we see from Corollary \ref{corH4}  that the $\cF_i$
also Poisson commute with the elements of $\fH$ on $M_\sigma^*$.
Moreover, the $\cF_i$  are identically zero on an open neighborhood of every point
outside $M_\sigma^*$, since every such point belongs to the complement of the closed set given by the
joint support of the functions  $\cF_i$.
As a result, $\{\cF_i, \cH\} = 0$ identically on $M$ for every $\cH \in \fH$, i.e.,
we proved that $\cF_i \in \fF$.
The upshot is that the exterior derivatives of the elements of $\fF$ span a $D$-dimensional subspace of
$T_{x_0}^*M$ at every point $x_0\in M_\sigma^*$, which proves the claim.
\end{proof}

The preceding lemmas complete the proof of  Theorem \ref{main}.
The next proposition characterizes the rank of the superintegrable system, $\ddim(\fH)$,
as the difference between the dimensions of the typical isotropy groups for $G$ acting
 on $\cJ(M)$ and on $M$.

\begin{prop}\label{propBJ}
The equality  $\ddim(\fH) = \dim(G_\sigma) - \dim(G_{x_0})$ holds,
where $G_{x_0}$ is an isotropy group of minimal dimension for the $G$-action on $M$ and
$G_\sigma$ is the isotropy group of the elements of the
principal open face $\sigma$ given by Theorem \ref{thmT9}.
\end{prop}
\begin{proof}
For any  $x\in M$ with isotropy group $G_x$,
the standard relation \cite{OR}  $\ann(\fg_x) = \Im (D\cJ(x))$ implies the equality
$ \dim(G\cdot x)=\rank(D\cJ(x)) $.
For $x\in M_\sigma$, we have  $\Im (D\cJ(x) )= \Im (D \cJ_\sigma(x))$.
By using that $\Sigma_\sigma$ (which contains $\cJ_\sigma(M_\sigma)$) is diffeomorphic to the product manifold $G/G_\sigma \times \sigma$,
$J_\sigma$  takes the form $\cJ_\sigma = (L_\sigma, \mu_\sigma)$
with a $G$-equivariant smooth map $L_\sigma: M_\sigma \to G/G_\sigma$.
This leads to
\be
\dim(G\cdot x) = \rank(D \cJ_\sigma(x)) = \dim(G/G_\sigma) + \rank (D \mu_\sigma(x)),
\ee
equivalently written as
\be
\rank (D\mu_\sigma(x)) = \dim(G_\sigma) - \dim (G_x), \qquad \forall x \in M_\sigma.
\ee
The rank is constant  over the dense open submanifold of $M_\sigma$ where $\dim(G_x)$ is minimal.
This constant  is  equal to the dimension of the generic $T_\sigma$-orbits,
since their tangent spaces are spanned by the Hamiltonian vector fields generated
by $\mu_\sigma$.
On account of \eqref{ddimH}, the claim follows.
\end{proof}

\subsection{Generalized action-angle coordinates}
\label{subsec32}

The action \eqref{GSact} of $T_\sigma$ on $M_\sigma^*$ is not necessarily free  (see Subsection \ref{subsec24}),
but the isotropy group is constant over $M_\sigma^*$
simply because $T_\sigma$ is Abelian.
Let us again choose an arbitrary point $x_0$ in $M_\sigma^*$ and denote $(T_\sigma)_{x_0}$ its isotropy group.
Then, define the factor group
\be
\hat T_\sigma := T_\sigma/ (T_\sigma)_{x_0}.
\label{H6}\ee
This is a torus of dimension
\be
\ell := \dim(M_\sigma^*) - \dim(M_\sigma^*/T_\sigma) = \dim(M) - D = \ddim(\fH),
\label{H7}\ee
which acts \emph{effectively on $M_\sigma$ and freely}
on $M_\sigma^*$.
It is easy to see that the induced torus action is Hamiltonian, and the corresponding momentum map is provided by
\be
\hat \mu_\sigma := \mu_\sigma - \mu_\sigma(x_0).
\label{H8}\ee
Notice that $\langle \mu_\sigma, X  \rangle$ is constant over $M_\sigma$ for every element $X$ from the Lie algebra $(\ft_{\sigma})_{x_0}$
of the isotropy group $(T_\sigma)_{x_0}$, because its Hamiltonian vector field vanishes on the connected symplectic manifold $M_\sigma$.
It follows that $\hat \mu_\sigma$ varies in the annihilator of the isotropy Lie algebra $(\ft_{\sigma})_{x_0}$ inside $\ft_\sigma^* \simeq \ft_\sigma$.
By standard linear algebra, this annihilator is naturally isomorphic to the dual space of the
Lie algebra $\hat \ft_\sigma \equiv \ft_\sigma/(\ft_\sigma)_{x_0}$
of the factor group $\hat T_\sigma$.

The  next theorem deals with
generalized action-angle coordinates.
Recall that $B_r^j$ denotes an open ball as defined before Lemma \ref{lemH5}.

\begin{thm}\label{thmaction}
With $\ell = \ddim(\fH)$ in \eqref{H7},
put $k:= \frac{1}{2}(\dim(M) - 2\ell)$. Choose $x_0\in M_\sigma^*$ and a basis $Y_1,\dots, Y_\ell$ of the integral lattice
of the torus $\hat T_\sigma$ \eqref{H6}.
Then,  there exists a $\hat T_\sigma$-invariant open neighborhood $\cU$ of $x_0$ in $M_\sigma^*$
and  a diffeomorphism
\be
\phi: \cU \to  \hat T_\sigma \times B^\ell_{r_1} \times B^{2k}_{r_2}
\label{H9}\ee
for some $r_1, r_2> 0$, so that $\phi(x_0) = (e, 0,0)$  and the following properties are satisfied.
First, the push-forward of the restriction of the symplectic form $\omega$ on $\cU$ can be written as
\be
\phi_* (\omega_{\vert \cU}) =\sum_{i=1}^\ell d I_i \wedge d \theta_i + \sum_{j =1}^{k} dp_j \wedge d q_j,
\label{H10}\ee
where $I_1,\dots, I_\ell$ and $q_1,\dots, q_k, p_1, \dots, p_k$ are coordinates on $B^\ell_{r_1}$ and $B^{2k}_{r_2}$.
Second, the $\theta_i$ are  angles parametrizing  the torus $\hat T_\sigma \simeq {\mathrm{U}}(1)^\ell$ \eqref{H7}
and the canonical conjugates of  the coordinate functions $\theta_i \circ \phi$  are given by the momentum map \eqref{H8}
of  the $\hat T_\sigma$-action according to
\be
I_i \circ \phi =   \langle \hat \mu_\sigma, Y_i \rangle.
\label{H11}\ee
Third, the action coordinates $I_i \circ \phi$ can be expressed in terms of  elements of $\fH_{\vert \cU}$
and the `transversal coordinates' $p_j\circ \phi, q_j \circ \phi$ can be expressed in terms of $\fF_{\vert \cU}$.
 \end{thm}
\begin{proof}
 Basically, the claims are consequences of the well-known generalized action-angle theorem  of Nekhoroshev \cite{Nek}.
More concretely,
the statements follow directly from \cite[Thm.~2.15 and Cor.~2.17]{FF} by using the fact that
the components of the momentum map $\hat \mu_\sigma$ \eqref{H8} provide action variables on $M_\sigma \subset M$ for the superintegrable system
$(M, \omega, \fH, \fF)$ in the sense of \cite[Defn.~1.2]{FF}.
\end{proof}

Superintegrable systems  on symplectic manifolds locally admit (near compact connected components of regular level sets)
generalized action-angle coordinates  in which the symplectic structure takes the form \eqref{H10} and the  superintegrable Hamiltonians
become functions of the action variables \cite{Nek}.
The key step in the construction of such coordinates is to find the action variables, which is in  general a very non-trivial problem.
In our case, we have identified the action variables as components of the momentum map of the GS torus action,
which are continuous functions on the whole phase space and are $C^\infty$  on the dense open
submanifold $M_\sigma$.

\subsection{Comparison with earlier results}
 \label{subsec33}

 Below, we compare our results with those  of
  Bolsinov and Jovanovi\'{c} \cite{BJ}.
  They applied the approach to superintegrability initiated in \cite{MF},
  whose central concept is the so-called `complete algebra of functions'.
  By definition \cite{BJ}, a complete algebra of functions $\cA$ on a symplectic manifold $(M,\omega)$
  is a linear subspace
  of the space of all smooth functions, which is closed under the
  Poisson bracket\footnote{Whether one also requires closure under pointwise multiplication of functions is optional.}
  and satisfies the condition
  \be
  \ddim(\cA) +\dind(\cA) = \dim(M).
  \label{H12}\ee
 Here, $\dind(\cA)$ (the \emph{differential index} of $\cA$) is the dimension of the kernel $K_x \subset A_x$
 of  the Poisson tensor restricted to
 $A_x \otimes A_x$ with $A_x:= \span\{ df(x)\mid f \in \cA\}$, for generic $x\in M$.
 Equivalently, for the points of a dense open subset of $M$, denoted $\reg(\cA)$, one has
 $\dim(A_x) =\ddim(\cA)$ and
 \be
 W_x:= (\omega_x^{-1})^\sharp(A_x) < T_x M
 \label{HJ13}\ee
 is a coisotropic  subspace with respect to the symplectic form.
 As defined in \cite{BJ},
 a Hamiltonian $\cH\in C^\infty(M)$ is called superintegrable if it admits a complete algebra
 of functions consisting of integrals of motion, such that $\ddim(\cA) < \dim(M)/2$.
 Generalized action-angle coordinates  exist
 around the connected components of the compact level surfaces of $\cA$ lying in $\reg(\cA)$,
and can be used  to characterize  the integral curves of the superintegrable Hamiltonians in the usual way
\cite{BJ, Nek}.

Bolsinov and Jovanovi\'{c}  \cite[Thm.~2.1]{BJ} proved that
\be
\cA:= C^\infty(M)^G + \cJ^* \left(C^\infty(\fg^*)\right)
 \label{H14}\ee
 is a complete algebra of functions.
 Notice that this is an infinite-dimensional Lie algebra.
 The proof given in \cite{BJ} is very direct,  and it works for certain Hamiltonian actions of non-compact connected Lie groups
 as well.
 They have also shown that
 \be
 \dind(\cA) = \dim(G_{\mu}) - \dim(G_x),
 \label{H15}\ee
 where $G_x$ and $G_\mu$ are isotropy groups associated with a generic $G$-orbit in $M$ and in $\cJ(M)$, respectively.
 It is also worth observing from the proof in \cite{BJ} that in our case $\reg(\cA)$ is the
 submanifold of  $M_\sigma$ \eqref{T21} consisting of the points whose $G$-isotropy groups have minimal dimension.

 Turning to the comparison, observe from  Definition \ref{defnI} that $\fF$ \eqref{I3} is a complete algebra of functions.
 With $\fH$ \eqref{I5} and $\cA$ \eqref{H14},  the relation
 \be
 \fH \subset \cA \subset \fF
 \label{H16}\ee
 is obviously valid.
 Moreover, by Proposition \ref{propBJ} and equation \eqref{H15},
 \be
 \ddim(\fF) = \ddim(\cA).
 \ee
 As a consequence, $\cA$ is a complete algebra of integrals of motion for every
 Hamiltonian $\cH \in \fH$,
and the generic integral curves of $\cH\in \fH$ can be characterized, in principle,
 also with the aid of generalized action-angle coordinates associated with $\cA$.

We conclude  that for compact Lie groups the two approaches lead to essentially equivalent results that
are also complementary to each other.
 In the direct approach of \cite{BJ} the complete algebra of constants of motion
is very explicit, while our approach shows that on the regular part the properties of the system
 are encoded by a torus action.
 We remark that the general equality $\dind(\cA) = \ddim(\fH)$ with $\fH$ in \eqref{I5}  was not stated in \cite{BJ,J}.

\section{Discussion}
\label{sec4}

We have shown that the Abelian Poisson algebra $\fH$ \eqref{I5} engenders
a superintegrable system in the sense of Definition \ref{defnI} for every Hamiltonian action of a compact connected Lie group $G$ associated
with a simple Lie algebra.
It is easy to see that (after small notational modifications) everything works for reductive compact Lie groups as well.
Our proof relied on the Hamiltonian torus action \eqref{GSact} that directly gives
action variables for these collective superintegrable systems.
This is important because action variables are often useful for semi-classical quantization and also because
the  momentum map \eqref{musi} may induce action variables
for many interesting superintegrable systems that can be obtained by Hamiltonian reduction
from systems covered by our framework.  For example, a family of spin Calogero--Sutherland models results
from reductions of cotangent bundles of compact symmetric spaces $G/K$ by using the action of $K$ inherited
from left-multiplications.
(See \cite{FP}, which focused on the non-compact case.)
 A part of the pertinent GS torus action is expected to descend to the
reduced phase space, leading to action variables of those  spin many-body models.
We  plan to elaborate this in a future publication.

Our work complements the earlier results obtained by Bolsinov and Jovanovi\'{c} in \cite{BJ}, where
the GS torus action was not used.
In fact, it has been one of the aims of the present paper to shed new light on the usefulness of this  tool,
which can also be incorporated into the Hamiltonian reduction approach
to  superintegrable systems based on Poisson--Lie groups
and quasi-Hamiltonian methods \cite{FF}.

Let us end with a  remark on a connection between Thimm's trick \cite{GS2,Lane,T} and superintegrability.
For this, consider a compact connected  Lie group $G$ with Lie algebra $\fg$ and let
\be
\fg_n  < \fg_{n-1} < \cdots < \fg_1:= \fg
\ee
be a chain of subalgebras with associated connected subgroups $G_i < G_1 := G$ for $i =1,\dots, n$.
If $G$ acts with momentum map $\cJ$ on a connected symplectic manifold $M$, then also the subgroups
act with corresponding momentum maps $\cJ_i = \pi_i \circ \cJ$, where $\pi_i: \fg^* \to \fg_i^*$
is the projection dual to the inclusion $\fg_i \to \fg$.
Taken together, the Abelian Poisson subalgebras $\fH_i := \cJ_i^* ( C^\infty(\fg_i^*)^{G_i})$ of $C^\infty(M)$ generate a `big' Abelian subalgebra $\fH$,
which was Thimm's basic observation \cite{T}.
By combining  Lie-theoretic (connectedness) results presented in \cite{Lane} with  a natural generalization of our proof of Theorem \ref{main},
it is not difficult to see that the `big'  $\fH$ and its
centralizer $\fF$ in $C^\infty(M)$ always satisfy the equality $\ddim(\fH) + \ddim(\fF) = \dim(M)$.
In general, this gives a superintegrable system, and  in very rare, celebrated cases \cite{GS2} a Liouville integrable one.
Of course, action variables are provided by the momentum map of the associated GS torus action for all these systems.

\bigskip
\subsubsection*{Acknowledgements}
I would like to thank  K.-H.~Neeb for very  helpful correspondence and useful suggestions.
I am also grateful to M.~Fairon for comments on the manuscript  and R.~Tr\'enyi for help with the figure.
The comments and questions of the anonymous referees are also much appreciated.
This  work was supported in part by National Research, Development and
Innovation Office of Hungary under the grant NKKP Advanced 152467.

\bigskip
\medskip





\newpage

\begin{thebibliography}{99}


    \setlength{\parskip}{0em}

\normalsize


 \bibitem{BJ}
 A.V.~Bolsinov and B.~Jovanovi\'{c},
 {\it Noncommutative integrability, moment map and geodesic flows}.
 Ann. Glob. Anal. and Geom. {\bf 23} (2003), 305-322;
 \href{ https://arxiv.org/abs/math-ph/0109031}{\tt arXiv:math-ph/0109031}


\bibitem{DK}
J.J.~Duistermaat and J.A.C.~Kolk,
Lie Groups. Universitext. Springer, 2000


\bibitem{FF}
L.~Feh\'er and M.~Fairon,
{\it  Integrable systems from Poisson reductions of generalized Hamiltonian torus actions.}
Nonlinearity {\bf 39} (2026) Article ID 075021;
 \href{https://arxiv.org/abs/2507.12051}{\tt arXiv:2507.12051}

\bibitem{FP}
L.~Feh\'er and B.G.~Pusztai,
{\it Spin Calogero models associated with Riemannian symmetric spaces of negative curvature.}
Nucl. Phys. B {\bf 751} (2006) 436-458;
\href{https://arxiv.org/abs/math-ph/0604073}{\tt arXiv:math-ph/0604073}

\bibitem{GS1}
V.~Guillemin and S.~Sternberg,
{\it The moment map and collective motion.}
Ann. of Phys. {\bf 127} (1980) 220-253


\bibitem{GS2}
V.~Guillemin and S.~Sternberg,
{\it On collective complete integrability according to the method of Thimm.}
Ergod. Th. and Dynam. Syst. {\bf 3} (1983) 219-230


\bibitem{GS3}
V.~Guillemin and S.~Sternberg,
{\it The Gelfand--Cetlin system and quantization of complex flag manifolds.}
J. Funct. Anal. {\bf 52} (1983) 106-128

\bibitem{HNP}
J.~Hilgert, K.-H.~Neeb and W.~Plank,
{\it Symplectic convexity theorems and coadjoint orbits.}
Compositio Math. {\bf 94} (1994) 129–180



\bibitem{J}
B.~Jovanovi\'c,
{\it Symmetries and integrability}.
Publ. Institut Math. {\bf 49} (2008) 1-36;
 \href{https://arxiv.org/abs/0812.4398}{\tt arXiv:0812.4398}


\bibitem{Lane}
J.~Lane,
{\it Convexity and Thimm's trick.}
Transform. Groups {\bf 23} (2018) 963-987;
\href{https://arxiv.org/abs/1509.07356}{\tt arXiv:1509.07356}


\bibitem{LMTW}
E.~Lerman, E.~Meinrenken, S.~Tolman and C.~Woodward,
{\it Non-Abelian convexity by symplectic cuts.}
Topology {\bf 37} (1998) 245-259;
\href{https://arxiv.org/abs/dg-ga/9603015}{\tt arXiv:dg-ga/9603015}



\bibitem{Mi}
P.W.~Michor, Topics in Differential Geometry.  Amer. Math. Soc., 2008



\bibitem{MPW}
W.~Miller~Jr, S.~Post and P.~Winternitz,
{\it Classical and quantum superintegrability with applications}.
J. Phys. A: Math. Theor. {\bf 46} (2013) Article ID 423001;
\href{https://arxiv.org/abs/1309.2694}{\tt arXiv:1309.2694}


\bibitem{MF}
A.S.~Mishchenko and A.T.~Fomenko,
{\it  Generalized Liouville method for integrating Hamiltonian systems}.
Funct. Anal. Appl. {\bf 12} (1978) 113-125

\bibitem{Nek}
N.N.~Nekhoroshev,
{\it  Action-angle variables and their generalizations}.
Trans. Moscow Math. Soc. {\bf 26} (1972) 180-197



\bibitem{OR}
J.-P.~Ortega and T.~Ratiu,
Momentum Maps and Hamiltonian Reduction. Birkh\"auser, 2004



\bibitem{R}
 N.~Reshetikhin,
{\it Degenerately integrable systems}.
J. Math. Sci. {\bf 213} (2016) 769-785;
\href{https://arxiv.org/abs/1509.00730}{\tt arXiv:1509.00730}


\bibitem{Sam}
H.~Samelson,
Notes on Lie algebras. Springer, 1990


\bibitem{T}
A.~Thimm,
{\it  Integrable geodesic flows on homogeneous spaces.}
Ergod. Th. and Dynam. Syst. {\bf 1} (1981) 495-517.

\bibitem{W}
C.~Woodward,
{\it Multiplicity free Hamiltonian actions need not be K\"ahler.}
Invent. Math. {\bf 131} (1998) 311–319;
\href{https://arxiv.org/abs/dg-ga/9506009}{\tt arXiv:dg-ga/9506009}

\bibitem{W2}
C.~Woodward,
{\it The classification of transversal multiplicity-free group actions.}
Ann. Glob. Anal. and Geom. {\bf 14} (1996) 3-42


\bibitem{Z}
 N.T.~Zung, {\it Torus actions and integrable systems.} pp. 289-328
 in: Topological Methods in the Theory of Integrable Systems. Camb. Sci. Publ., 2006;
 \href{https://arxiv.org/abs/math/0407455}{\tt arXiv:math/0407455}



\end{thebibliography}
\end{document}